\documentclass[11pt, reqno]{amsart}

\usepackage[utf8]{inputenc}
\usepackage{pgffor}
\usepackage{xcolor,soul}
\definecolor{hyperlink}{rgb}{0.7 0 0}
\usepackage[colorlinks=true,allcolors=hyperlink]{hyperref}
\usepackage{booktabs}
\usepackage{mathtools}
\usepackage{multicol}
\usepackage{float}
\usepackage{mathdots}
\usepackage{amsthm}
\usepackage{amssymb, mathabx}
\usepackage{graphicx}
\usepackage{cases}
\usepackage{array}
\usepackage{enumitem}
\usepackage[figureposition=top, skip=3pt]{caption}
\usepackage{subcaption}
\usepackage{tikz}
\usepackage{pgfplots}
\usepackage[compress]{natbib}
\usepackage{siunitx}

\makeatletter

\usepackage{mdframed}

\newmdenv[
  topline=false,
  bottomline=false,
  rightline=false,
  skipabove=\topsep,
  skipbelow=\topsep
]{leftrule}

\renewenvironment{proof}[1][\proofname]
    {\begin{leftrule}   
        \par        \normalfont \topsep6\p@\@plus6\p@        \trivlist
                        \item[\hskip\labelsep\itshape
              #1.]\ignorespaces
    }
    {\qed\endtrivlist\end{leftrule}}

\def\@seccntformat#1{\@ifundefined{#1@cntformat}   {\csname the#1\endcsname.\ }   {\csname #1@cntformat\endcsname}}\newcommand\paragraph@cntformat{\theparagraph\;}
\renewcommand\theparagraph{\textbullet\ \thesubsubsection(\alph{paragraph})}
\def\subsubsection{\@startsection{subsubsection}{3}  \z@\z@{-\fontdimen2\font}  {\normalfont\bfseries}}
\def\paragraph{\@startsection{paragraph}{4}  \z@\z@{-\fontdimen2\font}  {\normalfont\itshape}}

\DeclareCaptionSubType*{figure}
\newtheorem{proposition}{Proposition}[section]
\newtheorem{theorem}[proposition]{Theorem}
\newtheorem{corollary}[proposition]{Corollary}
\newtheorem*{theorem*}{Theorem}

\theoremstyle{definition}

\theoremstyle{remark}
\newtheorem{remark}[proposition]{Remark}

\renewcommand*\d{\mathop{}\!\d }

\newcommand{\customlabel}[2]{#2\def\@currentlabel{#2}\label{#1}}

\newcommand{\BS}{{\mathrm{BS}}}
\renewcommand{\d}{{\mathrm d}}

\makeatother

\calclayout

\begin{document}
\title{The Delta of a Variance Swap}

\author{ S\'ebastien Bossu and Sebastian Gaitan-Escarpeta\textsuperscript{*}}
\thanks{*\ UNC Charlotte Department of Mathematics and Statistics.  SB is the corresponding author, is indebted to Peter Carr for his interest and suggestions on this topic in 2017, and thanks Mehdi Sonthonnax at UBS for his direct proof of Theorem~\ref{thm:deltazero}, as well as Zhenyu Cui, Kailin Ding, Hao Lu and H. Shi for various past suggestions and contributions.  SB is partially supported by funds provided by the U.S. Department of Defense/War (Grant Award H98230-26-1-0016) and by The University of North Carolina at Charlotte. The United States Government is authorized to reproduce and distribute reprints notwithstanding any copyright notation herein.  }

\maketitle

\begin{abstract}
 We define the variance swap delta as the sensitivity of the price of variance to a change in underlying price.  We use Carr-Madan spanning formulas to analyze this sensitivity when the implied volatility smile curve may depend on the underlying price.  We show that the variance swap total delta is zero for the class of smile curves that are pure functions of (log) moneyness, which goes against the empirical observation that variance is up when the market is down.  We propose a simple modification of the smile to correct this issue.
     \end{abstract}
\renewcommand{\arraystretch}{1.5}

\section{Introduction}

The implied volatility smile or skew has been actively studied by academics and practitioners alike since the 1987 crash as a market adjustment to the Black-Scholes-Merton option pricing model.  In order to correctly price over-the-counter vanilla options and many exotic derivatives such as binary options or variance swaps, a carefully curated smile model is crucial.  Given the implied volatilities of a finite number of options, we can distinguish three common smile modelling approaches in ascending order of technical sophistication:
\begin{enumerate}[leftmargin=*]
 \item Interpolation and extrapolation of the smile, for example with cubic splines \citep[see, e.g.,][]{Fengler2009};
 \item Parametric model fitting, such as the Stochastic Volatility Inspired (SVI) model \citep{Gatheral2004};
 \item Option pricing model calibration, such as the Stochastic Alpha Beta Rho (SABR) model \citep{Hagan2002}, the Heston  model \citep{Heston1993}, or the class of models known as Stochastic Volatility with Jumps (SVJ) \citep[see, e.g.,][]{Bates1996,Kou2002}, to name a few.
\end{enumerate}

The first two approaches are attractive in terms of computational speed, but dealing with no-arbitrage conditions for the resulting option prices can be treacherous.  In contrast, the third approach is arbitrage-free by construction, but numerical tractability is often an issue.  In this paper, we only consider the first two approaches.

It is important for a smile model to incorporate some practical rules about the  behavior of implied volatility as the underlying price  moves.  Practitioners tend to roughly agree on the ``sticky delta rule'' (see, e.g., \cite[p.~18]{Bossu2014} and \cite{Derman1999}), typically resulting in a higher total option delta than the Black-Scholes delta by way of the chain rule.  However, in practice, there is general agreement that optimal hedge ratios tend to be lower than the Black-Scholes delta \citep[see, e.g.,][]{Hull2017}.

A key aspect of designing a smile model is to predict  what the correct implied volatility should be as the underlying price moves, and thus what the correct option delta should be. In this paper, we study the resulting sensitivity of the fair value of variance, that is, the sensitivity of the \emph{fair strike} of the variance swap to changes in underlying asset price.  Emprirically, one would expect such sensitivity to be negative, i.e. as the underlying price drops, the variance swap should be more expensive.
Specifically, the fair value of variance is famously given by the \citet{Carr2001} spanning formula \citep[see also][]{Demeterfi1999}:
\begin{equation}\label{eq:fairvar}
 V^* = \frac2T\int_{0}^{F} \frac{\d K}{K^2} \,P(K) + \frac2T\int_{F}^{\infty} \frac{\d K}{K^2} \,C(K),
\end{equation}
where $V^*$ is the fair value of annualized variance to maturity $T$, $F$ is the underlying forward price, and $C(K)$ and $P(K)$ are \emph{undiscounted} European call and put values with strike $K$.  Volatility indexes such as the VIX and VSTOXX are computed using a discretized version of this formula.

A straightforward application of Leibniz's integral rule together with put-call parity then shows that the \textbf{variance swap delta} –-- the change in the fair value of variance resulting from an infinitesimal change in underlying forward price---can be represented as the weighted sum of call and put deltas:
\begin{equation}
 \frac{\d V^\ast}{\d F}=\frac2T\int_{0}^{F}\frac{\d K}{K^2} \,{\frac{\d P}{\d F}\left(K,F\right)}+\frac2T\int_{F}^{\infty}\frac{\d K}{K^2} \,{\frac{\d C}{\d F}\left(K,F\right)}.
\end{equation}
An equivalent concept is the price elasticity of the annualized variance swap rate $\sigma^* = \sqrt{V^*}$, which is defined as:
\[
    \lambda = \frac{F}{\sigma^\ast}\frac{\d\sigma^\ast}{\d F}=\frac{1}{2}\frac{F}{V^\ast}\frac{\d V^\ast}{\d F}.
\]

In this paper, we show that the variance swap total delta $\d V^*/\d F$ is identically zero for the very general class of smile models that are pure functions of (log) moneyness, which notably includes the SVI model.  This counter-intuitive property clearly indicates that this class of models does not take into account the empirical co-dependence between underlying price and volatility. We also propose a straightforward correction of the smile to produce a (locally) constant price elasticity of the variance swap rate at any desired level $\lambda$.  We show that for $\lambda<0$ the resulting ``power smile'' total delta is lower than the ``powerless'' delta, which may help produce better hedge ratios.

\subsection{Standing notations and assumptions}

For ease of exposition, throughout this paper all prices and values are taken to be forward to time horizon $T$, and all options expire at time $T$.  We respectively write $P_\BS(F,K,\varsigma)$ and $C_\BS(F,K,\varsigma)$ for the \emph{undiscounted} Black-Scholes put and call prices with strike price $K$, forward price $F$, and volatility $\varsigma$, that is,
\begin{align*}
    & C_\BS(F,K,\varsigma) = FN(d_1) - KN(d_2), \quad 
    P_\BS(F,K,\varsigma) = KN(-d_2) - FN(-d_1),
\end{align*}
where $ d_{1,2} = \frac{\ln(F/K)}{\varsigma\sqrt T} \pm \frac12 \varsigma\sqrt T $.  By put-call parity, the Black-Scholes vega is the same for calls and puts and is denoted as $\mathcal V = \partial P_\BS/\partial \varsigma = \partial C_\BS/\partial \varsigma$.  The standard normal density function is denoted as $n(z)$.  The implied volatility smile $\hat\sigma(K,F)$ is assumed to be a smooth (twice continuously differentiable) function of the strike price $K$ and forward price $F$, and such that the Carr-Madan spanning formula for fair variance is well defined and the conditions of differentiability under the integral sign are met.

\subsection{Related literature and organization of the paper} 

In the practitioner literature, \cite{coulombe:2009} derive a formula similar to equation~\eqref{eq:BS-delta-x}. \cite{eid:2011} examines the skew delta empirically, and concludes that it is zero (``A variance swap that did not yet strike [...] is only a pure variance product and
should not be hedged with any offsetting delta on the spot'').  A popular industry paper that provides a trader perspective on variance swap was also circulated by \cite{BossuStrasserGuichard2005}.

In the academic literature, the sensitivity of volatility to a change in the underlying price remains a popular topic of research. \cite{ContDaFonseca2002} study empirically the dynamics of implied volatility surfaces and their dependence on the underlying index, while \cite{DaglishHullSuo2007} examine alternative rules for the evolution of the volatility surface, including sticky-strike and sticky-delta dynamics. The dynamics of the implied volatility surface have also been studied in the context of stochastic volatility models and their applications to SPX and VIX derivatives; see, for example, \cite{Papanicolaou2022}. These contributions motivate our focus on the dependence of the implied volatility smile on the underlying forward price, which is the key determinant of the total delta of a variance swap.  More broadly, in recent literature the significance of the Carr-Madan spanning formula has been further underscored by the multi-asset extensions of \cite{cui-xu:2022}, and \cite{BossuCarrPapanicolaou2021,BossuCarrPapanicolaou2022,Bossu2022,BossuCrepeyNguyen2025}.

This paper is organized as follows.  In Section~\ref{sec:totalvsbs}, we prove that the variance swap total delta is zero for the class of smile models that are pure functions of log-moneyness, and we derive formulas for the gap between the Black-Scholes and total variance swap delta.  In Section~\ref{sec:powersmile}, we introduce the  ``power smile'', which is a simple correction to produce a nonzero total delta. Finally, in Section~\ref{sec:numerical}, we use historical market data to propose estimates of $\lambda$ for various maturities, and we compare the option deltas generated by the power smile versus other methods.  

\section{Total versus Black-Scholes delta of a variance swap}\label{sec:totalvsbs}

For fixed maturity $T$, let $\hat\sigma(K, F)$ be the smile as a smooth function of the strike price $K$ and the forward price $F$, and define $\sigma(x, F) = \hat\sigma(F e^x, F)$ as the smile represented as a function of log-moneyness $x = \ln(K/F)$ and forward price $F$.  The fair value of annualized variance is then given as
\begin{align}
    V^*(F)
    & =\frac{2}{T}\int_0^F \frac{\d K}{K^2} \, P_\BS\!\left(F, K,
   \hat\sigma(K,F)\right)
   +\frac{2}{T}\int_F^{+\infty} \frac{\d K}{K^2} \, C_\BS\!\left(F, K,
   \hat\sigma(K,F)\right)
   \label{eq:varswap-carrmadan-K}
   \\
   & = \frac{2}{T}\int_0^F \frac{\d K}{K^2} \, P_\BS\!\left(F, K,
   \sigma\!\left(\ln\frac KF,F\right)\!\right)
   +\frac{2}{T}\int_F^{+\infty} \frac{\d K}{K^2} \, C_\BS\!\left(F, K,
   \sigma\!\left(\ln\frac KF,F\right)\!\right).
   \label{eq:varswap-carrmadan-K-bis}
\end{align}

\subsection{Total delta}
Following the above definitions, the total delta of a variance swap is simply defined as the sensitivity $\d V^*/\d F$.  The proposition below shows that it may be calculated as a weighted sum of the smile sensitivity to changes in $F$.

\begin{proposition}\label{prop:totaldelta}
    The variance swap total delta is given by
    \[
        \frac{\d V^*}{\d F}=\frac{2}{\sqrt T}\int_{-\infty}^{+\infty} n(d_2(x, F)) \frac{\partial\sigma}{\partial F}(x,F)\,\d x,
\]
where $d_2(x, F) \coloneqq \frac{-x}{\sigma(x, F)\sqrt{T}}-\frac{1}{2}\sigma(x, F)\sqrt{T}$.
\end{proposition}
\begin{proof}
Starting from equation~\eqref{eq:varswap-carrmadan-K-bis} and substituting $K=Fe^x$ (with $x$ held fixed while differentiating in $F$), so that $\d K/K^2 = e^{-x}\,\d x/F$,
\[
V^*(F) = \frac{2}{T}\int_{-\infty}^0 \frac{e^{-x}}{F}\,P_\BS\bigl(F,Fe^x,\sigma(x,F)\bigr)\,\d x + \frac{2}{T}\int_0^{+\infty}\frac{e^{-x}}{F}\,C_\BS\bigl(F,Fe^x,\sigma(x,F)\bigr)\,\d x.
\]
By homogeneity of $P_\BS$ and $C_\BS$ in $(F,K)$, we have $P_\BS(F,Fe^x,\varsigma)=F\,P_\BS(1,e^x,\varsigma)$ and likewise for $C_\BS$. Substituting into the above and simplifying,
\begin{equation}\label{eq:v*}
    V^*(F) = \frac{2}{T}\int_{-\infty}^0 e^{-x}\,P_\BS\bigl(1,e^x,\sigma(x,F)\bigr)\,\d x + \frac{2}{T}\int_0^{+\infty}e^{-x}\,C_\BS\bigl(1,e^x,\sigma(x,F)\bigr)\,\d x.
\end{equation}
Differentiating under the integral sign by the chain rule, and using the fact that the Black-Scholes vega is identical for calls and puts to merge integrals,
\begin{equation}\label{eq:dvdf}
    \frac{\d V^*}{\d F} = \frac{2}{T}\int_{-\infty}^{+\infty} e^{-x}\,\mathcal V\bigl(1,e^x,\sigma(x,F)\bigr)\,\frac{\partial\sigma}{\partial F}(x,F)\,\d x.
\end{equation}
Substituting $\mathcal V(F,K,\varsigma)=Kn(d_2)\sqrt T$ for $F=1$, $K=e^x$ and then simplifying gives 
\[
\frac{\d V^*}{\d F} = \frac{2\sqrt T}{T}\int_{-\infty}^{+\infty} n(d_2(x,F))\,\frac{\partial\sigma}{\partial F}(x,F)\,\d x = \frac{2}{\sqrt T}\int_{-\infty}^{+\infty} n(d_2(x,F))\,\frac{\partial\sigma}{\partial F}(x,F)\,\d x,
\]
as claimed. \end{proof}

The following theorem shows that the class of smiles that are pure functions of log-moneyness is characterized by a variance swap total delta of zero.

\begin{theorem}\label{thm:deltazero}
Provided that the smile $\sigma(x,F)$ is monotonic in $F$, it is independent from $F$ if and only if the fair value of variance is independent from $F$, i.e. for all $h > -F$,
\[
    \sigma(x, F + h) = \sigma(x, F) \text{ for any }x\in\mathbb R
    \qquad\text{if and only if}\qquad
    V^*(F + h) = V^*(F).
\]
\end{theorem}
\begin{corollary}
 If the smile is a pure function of (log) moneyness $x = \ln(K/F)$, the variance swap total delta is identically zero, i.e. $\d V^* / \d F = 0$.
\end{corollary}
\begin{remark}
    The assumption that $\sigma(x,F)$ be monotonic in $F$ is reasonable in practice.  In equities, for fixed log-moneyness $x$, one typically expects a higher implied volatility in a bear market when $F$ is down, and a lower implied volatility in a bull market when $F$ is up. This behavior is not expected to reverse for some values of $F$, thereby ensuring monotonicity. Note that monotoncity in $F$ says nothing about the shape of the smile curve as a function log-moneyness $x = \ln(K/F)$, which is typically decreasing up to some strike near the money, and then increasing.
\end{remark}
\begin{proof}[Proof of Theorem~\ref{thm:deltazero}] $(\Rightarrow)$~Let $h > -F$.  Evaluating equation~\eqref{eq:varswap-carrmadan-K-bis} at $F + h$,
\begin{align*}
V^*(F+h)  =  \frac{2}{T}\int_0^{F+h} & \frac{\d K}{K^2} \, P_\BS\!\left(F+h, K,
\sigma\!\left(\ln\frac{K}{F+h}, F+h\right)\right)                                                                   \\
                                    & +
\frac{2}{T}\int_{F+h}^{+\infty} \frac{\d K}{K^2} \, C_\BS\!\left(F+h, K,
\sigma\!\left(\ln\frac{K}{F+h},F+h\right)\right).
\end{align*}
By linear substitution $K \mapsto K \frac{F}{F+h} $,
\begin{align*}
    V^*(F+h)
    & = \frac{2}{T}\int_0^F
        \frac{\d K}{K^2} \left(\frac{F}{F+h}\right)^2 P_\BS\!\left(F+h, K
        \frac{F+h}{F}, \sigma\!\left(\ln\frac{K}{F},F+h\right)\right) \frac{F+h}{F}
 \\
    & \quad +\frac{2}{T}\int_F^{+\infty}
    \frac{\d K}{K^2} \left(\frac{F}{F+h}\right)^2
    C_\BS\!\left(F+h, K \frac{F+h}{F},
    \sigma\!\left(\ln\frac{K}{F},F+h\right)\right) \frac{F+h}{F}.
\end{align*}
Simplifying and using the assumption that $\sigma(x, F+h) = \sigma(x,F)$,
\begin{align}
V^*(F+h) &= \frac{2}{T}\int_0^F \frac{\d K}{K^2} \, \frac{F}{F+h}
P_\BS\!\left(F+h, K \frac{F+h}{F}, \sigma\!\left(\ln\frac{K}{F},F\right)\right) \nonumber \\
 & \qquad  +\frac{2}{T}\int_F^{+\infty} \frac{\d K}{K^2} \,
\frac{F}{F+h} C_\BS\!\left(F+h, K \frac{F+h}{F},
\sigma\!\left(\ln\frac{K}{F},F\right)\right).
\label{eq:Fplush}
\end{align}
By homogeneity of the Black–Scholes pricing formulas, for any scalar $\lambda>0$, we have $\lambda P_\BS(F, K, \varsigma)=P_\BS(\lambda F, \lambda K, \varsigma)$ and similarly for $C_\BS$.  Substituting into equation~\eqref{eq:Fplush} with $\lambda=\frac{F}{F+h}$ and simplifying,
\[
V^*(F+h)=\frac{2}{T}\int_0^F \frac{\d K}{K^2} \, P_\BS\!\left(F, K,
\sigma\!\left(\ln\frac{K}{F},F\right)\right)
+\frac{2}{T}\int_F^{+\infty} \frac{\d K}{K^2} \, C_\BS\!\left(F, K,
\sigma\!\left(\ln\frac{K}{F},F\right)\right)
\]
which is $V^*(F)$ as claimed.

\smallskip
\noindent $(\Leftarrow)$ By equation \eqref{eq:dvdf}, together with the assumption that $V^*(F)$ is constant, we obtain for every $F>0$:
\[
    0=\frac{\d V^*}{\d F}=\frac{2}{T}\int_{-\infty}^{+\infty} e^{-x}\mathcal{V}\bigl(1,e^x,\sigma(x,F)\bigr) \frac{\partial\sigma}{\partial F}(x,F)\,\d x.
\]
Because $\sigma(x, F)$ is assumed to be monotonic in $F$, and because the factor $e^{-x}\mathcal{V}\bigl(1,e^x,\sigma(x,F)\bigr)$ is always positive, the integrand above does not change sign.  Since the integral vanishes, the integrand must also be identically zero, whence
\[
\frac{\partial\sigma}{\partial F}(x,F)=0 \qquad \text{for all } F>0,\; x\in\mathbb{R},
\]
thereby proving that $\sigma(x,F)$ does not depend on $F$.
\end{proof}

As explained in introduction, this is a counter-intuitive result, as one expects volatility to increase when the underlying price decreases. An alternative proof is provided in Appendix~\ref{sec:appendixa}. In Section~\ref{sec:svitest}, we verified this theoretical prediction numerically for the SVI implied volatility smile model calibrated to market data.

\subsection{Variance swap Black-Scholes delta}

As an alternative to the total delta, it may instead be of interest to compute the variance swap Black-Scholes delta which is defined as
\begin{align}
    \delta_{\BS }
    & = \frac2T\int_{0}^{F}{\frac{\d K}{K^2}}\,{\frac{\partial P_\BS}{\partial F}\left(F,K,\hat\sigma(K,F)\right)}+\frac2T\int_{F}^{\infty}{\frac{\d K}{K^2}}\,{\frac{\partial C_\BS}{\partial F} \left(F,K,\hat\sigma(K,F)\right)}
    \label{eq:BS-delta}
    \\
    & = \frac2T\int_{0}^{F}{\frac{\d K}{K^2}}\,{\frac{\partial P_\BS}{\partial F}\!\left(F,K,\sigma\!\left(\ln\frac KF,F\right)\right)}+\frac2T\int_{F}^{\infty}{\frac{\d K}{K^2}}\,{\frac{\partial C_\BS}{\partial F}\!\left(F,K,\sigma\!\left(\ln\frac KF,F\right)\right)}
    \label{eq:BS-delta-log}
\end{align}
where $\frac{\partial P_\BS}{\partial F}, \frac{\partial C_\BS}{\partial F}$ respectively denote the Black-Scholes put and call deltas.
\begin{proposition} \label{prop:varswap-delta-gap}
  The gap between the Black-Scholes and total variance swap deltas is given by
  \begin{align}      
    \delta_{\BS } - \frac{\d V^*}{\d F}
    & = -\frac 2T\int_{-\infty}^\infty \frac{\d K}{K^2}\,\mathcal{V}(F,K,\hat\sigma(K,F))\frac{\partial\hat\sigma}{\partial F}(K,F)
    \label{eq:BS-total-gap} \\
    & = \frac{2}{\sqrt{T}}\int_{-\infty}^{\infty} n\big(d_2(x,F)\big)
\left[\frac{1}{F}\frac{\partial \sigma}{\partial x}(x,F) - \frac{\partial \sigma}{\partial F}(x,F)\right]\d x,
    \label{eq:bs-delta-x}
  \end{align}
where, as before, $\sigma(x,F) \coloneqq \hat\sigma(Fe^x,F)$ is the smile as a function of log-moneyness $x=\ln(K/F)$, and $d_2(x,F) \coloneqq \frac{-x}{\sigma(x,F)\sqrt{T}}-\tfrac{1}{2}\sigma(x,F)\sqrt{T}$.
\end{proposition}
\begin{corollary}
    If the smile $\sigma(x,F)$ only depends on (log) moneyness $x$, the variance swap Black-Scholes delta simplifies to
  \begin{equation}  \label{eq:BS-delta-x}
   \delta_{\BS } = \frac 2{F\sqrt{T}}\int_{-\infty}^\infty n\!\left(d_2(x)\right)\sigma'(x)\,\d x,
  \end{equation}
  wherein we dropped the dependence on $F$ in $\sigma(x)$ and $d_2(x)$.
\end{corollary}
\begin{remark}
    If the smile $\sigma(x)$ is a pure downward-sloping function of (log) moneyness, then the total delta is zero and, by equation~\eqref{eq:bs-delta-x}, the Black-Scholes delta is negative and thus lower than the total delta.  If, on the other hand, the smile $\hat\sigma(K,F)$ is decreasing in $F$ (i.e. goes up when the market $F$ goes down, consistently with empirical observation), by equation~\eqref{eq:BS-total-gap} the gap is positive and the Black-Scholes variance swap delta is higher than the total delta.
\end{remark}
\begin{remark}
    In general, $n(d_2(x)) = e^{-[d_2(x)]^2/2}/\sqrt{2\pi}$ is not a probability density because the implied volatility $\sigma(x)$ appears in $d_2(x)$ and varies with the integration variable $x$.  One notable exception is when the smile $\sigma(x) \equiv \sigma$ is flat, in which case $n(d_2(x))/\sigma\sqrt T$ is a density, but observe that in this case all measures of variance swap delta vanish.  \end{remark}
\begin{proof}[Proof of Proposition~\ref{prop:varswap-delta-gap}]
Apply the Leibniz rule to $V^*(F)$ as defined in equation~\eqref{eq:varswap-carrmadan-K}. The integrals there have a moving boundary at $K=F$, which contributes boundary terms $\frac{2}{F^2}[P_\BS(F,F,\hat\sigma(F,F)) - C_\BS(F,F,\hat\sigma(F,F))]$ that vanish by put-call parity. Differentiating the integrands by the chain rule,
 \begin{align*}
  \frac{\d V^*}{\d F}
   & = \frac2T\int_0^F \frac{\d K}{K^2}
  \left[\frac{\partial P_\BS}{\partial F}(F,K,\hat\sigma(K,F)) + \frac{\partial\hat\sigma}{\partial F}(K,F)\, \mathcal{V}(F,K,\hat\sigma(K,F))\right] \\
   & \quad + \frac2T\int_F^\infty \frac{\d K}{K^2}
  \left[\frac{\partial C_\BS}{\partial F}(F,K,\hat\sigma(K,F)) + \frac{\partial\hat\sigma}{\partial F}(K,F)\, \mathcal{V}(F,K,\hat\sigma(K,F))\right].
 \end{align*}
Splitting each integrand into its two summands and rearranging,
\begin{align*}
 \frac{\d V^*}{\d F}
  =\ &\frac2T\int_0^F \frac{\d K}{K^2}\frac{\partial P_\BS}{\partial F}(F,K,\hat\sigma(K,F)) + \frac2T\int_F^\infty \frac{\d K}{K^2}\frac{\partial C_\BS}{\partial F}(F,K,\hat\sigma(K,F))\\
  & + \frac2T\int_0^F \frac{\d K}{K^2}\,\frac{\partial\hat\sigma}{\partial F}(K,F)\,\mathcal{V}(F,K,\hat\sigma(K,F))
  + \frac2T\int_F^\infty \frac{\d K}{K^2}\,\frac{\partial\hat\sigma}{\partial F}(K,F)\,\mathcal{V}(F,K,\hat\sigma(K,F)).
\end{align*}
By equation~\eqref{eq:BS-delta}, the first two terms correspond to $\delta_{\BS }$. By put-call parity  the vega $\mathcal{V}$ is identical for both puts and calls and the last two integrals merge into a single integral over $K\in(0,\infty)$:
\[
 \frac{\d V^*}{\d F} = \delta_{\BS } + \frac2T\int_0^\infty \frac{\d K}{K^2}\,\mathcal{V}(F,K,\hat\sigma(K,F))\,\frac{\partial\hat\sigma}{\partial F}(K,F),
\]
which is equation \eqref{eq:BS-total-gap} after rearranging.
Next, by change of variable $K = F e^x$ applied to equation~\eqref{eq:BS-total-gap}, together with homogeneity of vega,
\[
    \delta_{\BS } -\frac{\d V^*}{\d F} = -\frac2T\int_{-\infty}^\infty e^{-x}\,\mathcal{V}(1,e^x,\hat\sigma(Fe^x,F))\,\frac{\partial\hat\sigma}{\partial F}(F e^x,F).
\]
By the chain rule applied to $\hat\sigma(K,F) = \sigma(\ln(K/F),F)$,
\[
    \delta_{\BS } -\frac{\d V^*}{\d F} =\frac2T\int_{-\infty}^\infty e^{-x}\,\mathcal{V}(1,e^x,\hat\sigma(Fe^x,F))\left[ \frac1F\frac{\partial\sigma}{\partial x}(x,F) - \frac{\partial\sigma}{\partial F}(x,F)\right]. 
\]
Substituting $\mathcal{V}(1,e^x,\hat\sigma(Fe^x,F)) = e^x n(d_2(x,F))\sqrt T$ and simplifying, we obtain equation~\eqref{eq:bs-delta-x} as required.

\end{proof}

\section{Power Smile}\label{sec:powersmile}

\newcommand{\bolt}{  \begin{tikzpicture}[baseline=(X.base), x=0.8ex, y=0.8ex]
        \node[inner sep=0pt, outer sep=0pt] (X) at (0,0) {\vphantom{H}};
        \fill[black] (0.8, 1.8) -- (0.1, 0.9) -- (0.5, 0.9) -- (0.2, 0.0) -- (0.9, 0.9) -- (0.5, 0.9) -- cycle;
  \end{tikzpicture}}

As shown in Theorem~\ref{thm:deltazero}, any smile model $\sigma(x)$ that is a pure function of log-moneyness $x = \ln(K/F)$ produces a variance swap total delta of zero, which contradicts the empirical observation that implied volatility is up whenever the underlying price drops.  In this section, we propose the following simple adjustment to produce a constant price elasticity of the variance swap rate:
\begin{equation}\label{power_smile}
 \sigma\bolt{} (x,F) = \sigma(x) \left(\frac F{F_0}\right)^\lambda,
\end{equation}
where $F_0$ is the initial price, $F$ is the spot price, and $\lambda$ is a parameter corresponding to the price elasticity of the variance swap rate.  It is worth emphasizing that Equation \eqref{power_smile} is local in nature and requires values of $F$ to be sufficiently close to $F_0$.  In practice, bounding the multiplication factor $(F/F_0 )^\lambda$ within a range is required to avoid nonsensical results.

\begin{proposition}\label{prop:elasticity}
Under the power smile equation~\eqref{power_smile}, the following hold for $F \approx F_0$:
\begin{enumerate}[leftmargin=*,label=(\alph*)]
    \item The fair value of variance satisfies
     \[ V^*(F) \approx  \left(\frac{F}{F_0}\right)^{2\lambda} V^*(F_0).
     \]
    \item The variance swap total delta is
     \[ \frac{\d V^*}{\d F} \approx \frac{2\lambda}{F}\,V^*(F_0). \]
    \item The price elasticity of the variance swap rate $\sigma^* = \sqrt{V^*}$ is constant, i.e.
     \begin{equation}\label{eq:elasticity}
      \frac{F}{\sigma^*}\cdot\frac{\d\sigma^*}{\d F} \approx \lambda.
     \end{equation}
    \item For a call or put option price $f(F,K) \coloneqq f_\BS(F,K,\sigma\bolt{}(\ln \frac KF,F))$, the option power delta is 
    \begin{equation}    \label{eq:option-power-delta}
        \frac{\d f}{\d F} = \frac{\partial f_\BS}{\partial F} + \frac{F^{\lambda-1}}{F_0^\lambda}\left[ \lambda\,\sigma\!\left(\ln\frac KF\right) 
        - \sigma'\!\left(\ln\frac KF\right) \right] \frac{\partial f_\BS}{\partial\varsigma}. 
    \end{equation}
    \end{enumerate}
\end{proposition}
\begin{corollary}
    For $\lambda<0$, the option power delta at $F=F_0$  is lower than the ``powerless'' total delta $ \frac{\d}{\d F} f_\BS(F,K,\sigma(\ln \frac KF)) = \frac{\partial f_\BS}{\partial F} 
        - \frac1F \sigma'\!\left(\ln\frac KF\right) \frac{\partial f_\BS}{\partial\varsigma} $.
\end{corollary}
\begin{proof}[Proof of Proposition~\ref{prop:elasticity}]
(a)~Substituting $ \sigma\bolt{}(x,F) = \sigma(x)(F/F_0)^\lambda $ into equation \eqref{eq:carr-lee-polishchuk}, and factoring $(F/F_0)^{2\lambda}$ out of the integral,
\[
    V^*(F) = \left(\frac F{F_0}\right)^{2\lambda} \int_{-\infty}^{\infty}n\big(d_2{\!\!\bolt{}}(x, F)\big)\sigma^2(x)\Big(-\frac{\partial d_2{\!\!\bolt{}}}{\partial x}(x,F)\Big)\d x,
\]
wherein \( \displaystyle
  d_2{\!\!\bolt{}}(x,F) \coloneqq \frac{-x-(F/F_0)^{2\lambda}\sigma^2(x) T/2}{(F/F_0)^{\lambda}\sigma(x)\sqrt T}.
\)
In the limit as $F \to F_0$, we have
\(
  \displaystyle  d_2{\!\!\bolt}(x,F) \to d_2(x) \coloneqq \frac{-x-\sigma^2(x) T/2}{\sigma(x)\sqrt T}.
\)
Taking the limit under the integral,
\[
    \int_{-\infty}^{\infty}n\big(d_2{\!\!\bolt{}}(x, F)\big)\sigma^2(x)\Big(-\frac{\partial d_2{\!\!\bolt{}}}{\partial x}(x,F)\Big)\d x
    \xlongrightarrow[F\to F_0]{}
    \int_{-\infty}^{\infty}n\big(d_2(x)\big)\sigma^2(x)\big(-d_2'(x)\big)\d x
    = V^*(F_0).
\]
Hence $ V^*(F) \approx \left(\frac F{F_0}\right)^{2\lambda} V^*(F_0)$ for $F \approx F_0$ as claimed.

\noindent (b)~Differentiating the local approximation of part~(a) directly with respect to $F$,\[
 \frac{\d V^*}{\d F} \approx \frac{\d}{\d F}\left[\left(\frac F{F_0}\right)^{2\lambda}V^*(F_0)\right]
 = \frac{2\lambda}{F}\left(\frac F{F_0}\right)^{2\lambda}V^*(F_0)
 \approx \frac{2\lambda}{F}\,V^*(F_0).
\]

\noindent (c)~By the chain rule, $ \displaystyle \frac{\d\sigma^*}{\d F} = \frac{1}{2\sigma^*}\frac{\d V^*}{\d F} $.  By part (b), for $F \approx F_0$,
\[
\frac{\d\sigma^*}{\d F} \approx \frac{1}{2\sigma^*}\cdot\frac{2\lambda}{F}V^*(F_0) = \frac{\lambda\sigma^*}{F},
\]
thereby proving that $\frac{F}{\sigma^*}\cdot\frac{\d\sigma^*}{\d F} \approx \lambda$.

\noindent (d)~By the chain rule,
$ \displaystyle
\frac{\d f_{\BS}}{\d F} = \frac{\partial f_{\BS}}{\partial F}+  \frac{\partial f_\BS}{\partial\varsigma} \, \frac{ \d \sigma \bolt{}}{\d F},
$
and \[
\frac{\d\sigma \bolt{}}{\d F}
= \frac{\partial \sigma{\bolt{}}}{\partial x}  \frac {\partial x }{\partial F} + \frac{\partial \sigma{\bolt{}}}{\partial F}
= -\frac{\sigma'(x)}{F} \left( \frac{F}{F_0}\right)^\lambda + \sigma(x) \frac{\lambda F^{\lambda-1}}{F_0^\lambda}.
\]
Substituting and simplifying yields the desired result. 
\end{proof}

\section{Numerical results}\label{sec:numerical}

\subsection{Numerical validation of Theorem~\ref{thm:deltazero}}\label{sec:svitest}
To numerically validate the results presented in Theorem~\ref{thm:deltazero}, we calibrated the SVI model \citep{Gatheral2004} to the SPX option chain observed on August 20, 2025 at 12:36pm EST, for the September 19, 2025 expiry ($T=0.082192$ years, i.e. about 1 month). The spot price was $S_0=6395.56$, with dividend yield $q=1.6\%$ and risk-free rate $r=4.8\%$, giving a forward price $F_0=S_0e^{(r-q)T}=6412.45$. 

Writing the SVI curve as a function of log-moneyness as
\[
 \sigma_{\mathrm{SVI}}(x) = \sqrt{a + b \left[ \rho\, (x-m) + \sqrt{(x - m)^2 + s^2} \right]},
\]
we determine the parameters $a$, $b$, $\rho$, $m$, and $s$ that minimize the total-variance squared error objective
\[
 \sum_{i=1}^{n} \left[ a+b\Big(\rho\,(x_i-m)+\sqrt{(x_i-m)^2+s^2}\Big) - \sigma_{\mathrm{Market},i}^{2} \right]^2,
\]
subject to $|\rho|<0.999$, $0<b<2$, $0.0001<s<2$, $|a|<1$, $|m|<1$.
The parameters were obtained via constrained least-squares optimization with 60 random restarts to mitigate the non-convexity of the SVI objective, retaining the best feasible fit under the no-arbitrage constraints. Table~\ref{tab:calibrated_params} reports the calibrated SVI parameters. Figure~\ref{fig:svifit} plots the calibrated SVI curve (blue) against the quoted SPX implied volatilities (red) as a function of log-moneyness; the close visual agreement between the two, together with the volatility relative RMSE of $5.2\%$, indicates that the SVI parametrization captures the shape of the market smile well across the full range of quoted strikes.

\begin{table}[h]
\centering
\begin{tabular}{>{\bf}lccccc}
\toprule
Parameter & $a$ & $b$ & $\rho$ & $s$ & $m$ \\
\midrule
Value & $-0.065314$ & $0.135658$ & $0.102604$ & $0.490019$ & $0.111212$ \\
\bottomrule
\end{tabular}
\caption{SVI calibration for the SPX smile on August 20, 2025.}
\label{tab:calibrated_params}
\end{table}

After calibrating the SVI parameters, we used numerical integration to compute the fair price of variance \( V^*(F) \) as a function of the forward price $F$:
\[
 V^*(F) = \frac{2}{T} \int_0^F \frac{\d K}{K^2} \,P_\BS\left(F, K, \sigma_{\mathrm{SVI}}\left(\ln\frac{K}{F}\right)\right) + \frac{2}{T} \int_F^{+\infty} \frac{\d K}{K^2} \, C_\BS\left(F, K, \sigma_{\mathrm{SVI}}\left(\ln\frac{K}{F}\right)\right).
\]
The resulting plot of $V^*(F)$, shown in Figure~\ref{fig:plot1}, displays a completely flat curve at $V^*\approx0.026303$ across $F\in(0,10000)$, with no discernible numerical noise. This corresponds to a fair variance swap rate $\sqrt{V^*}\approx16.21\%$ which is broadly consistent with the VIX which ranged between 15.57 and 17.19. This behavior confirms the theoretical result of Theorem~\ref{thm:deltazero}: under a volatility smile model that depends on log-moneyness only, the fair price of variance \( V^* \) is independent of the forward price.
\vspace{0cm}
\begin{figure}[H]
    \centering
    \includegraphics[width=\textwidth]{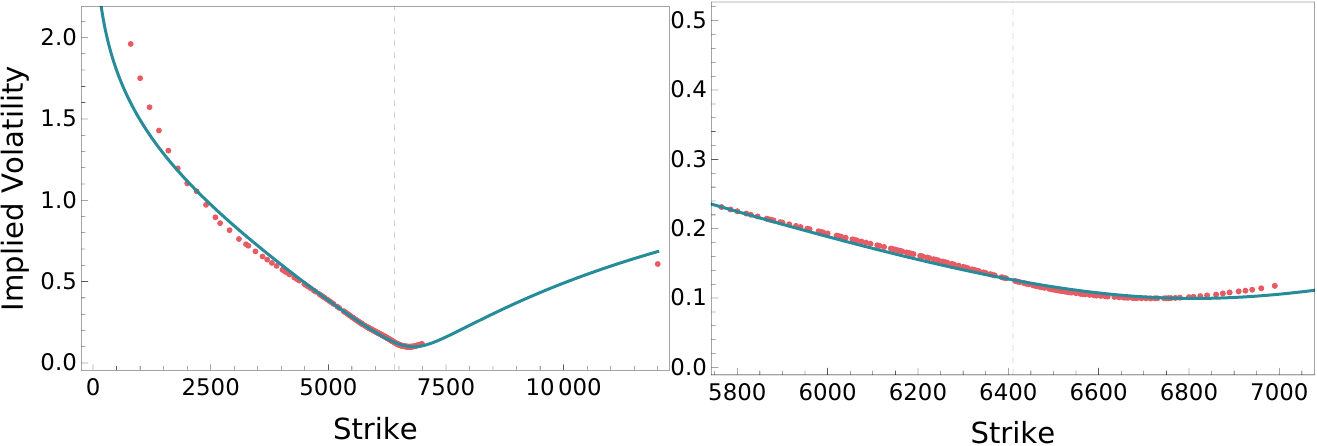}
    \caption{Calibrated SVI implied volatility smile (blue) against market implied volatilities (red) as a function of strike, for the SPX smile of August 20, 2025. The right panel is a detailed view for stikes near the forward price.}
    \label{fig:svifit}
\end{figure}

\begin{figure}[h]
 \centering
 \includegraphics[scale=0.7]{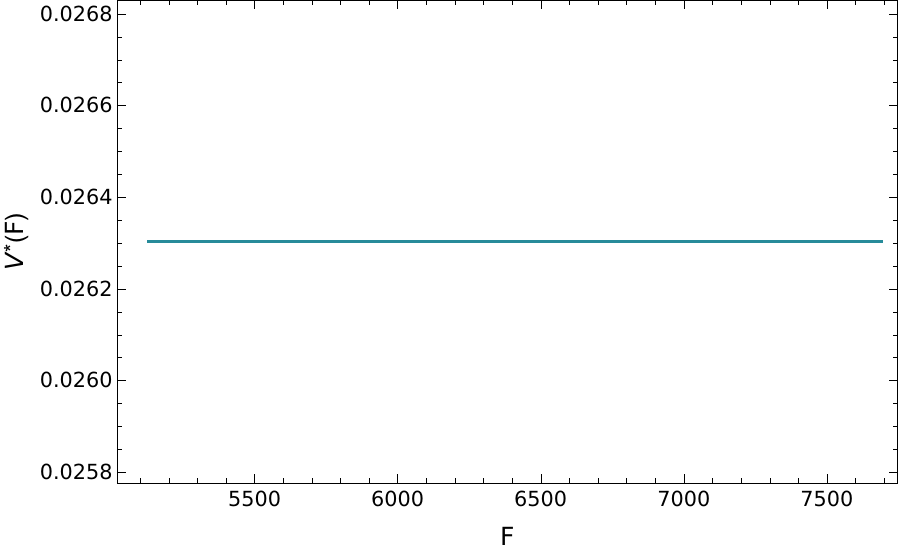}
 \caption{Plot of \( V^*(F) \) using the SVI volatility calibrated to the SPX smile of August 20, 2025 ($T=1$ month, September 19, 2025 expiry).}
 \label{fig:plot1}
\end{figure}

\subsection{Estimating $\lambda$}

Using the results from previous sections, we can estimate the value of
$\lambda$ from historical data. By Proposition~\ref{prop:elasticity}(c), the elasticity relation \eqref{eq:elasticity} can be written in return form as
\(
\frac{\d\sigma^*}{\sigma^*}=\lambda\,\frac{\d F}{F},
\)
i.e., the instantaneous return of the variance swap rate $\sigma^*$ is $\lambda$ times the instantaneous return of the underlying forward $F$. This suggests estimating $\lambda$, for each maturity, using a zero-intercept linear regression $Y = \lambda X + \varepsilon$ of variance swap rate returns $Y$ against index returns $X$. We used historical data of S\&P 500 index and variance swaps from 2008 to 2025.  Figure \ref{fig:lambdaregression} shows the regression  variance swap returns against index returns at different maturities, and Table \ref{t1} shows the estimates of $\lambda$ for various maturities.

\begin{figure}[p]
 \centering
 \includegraphics[scale=1]{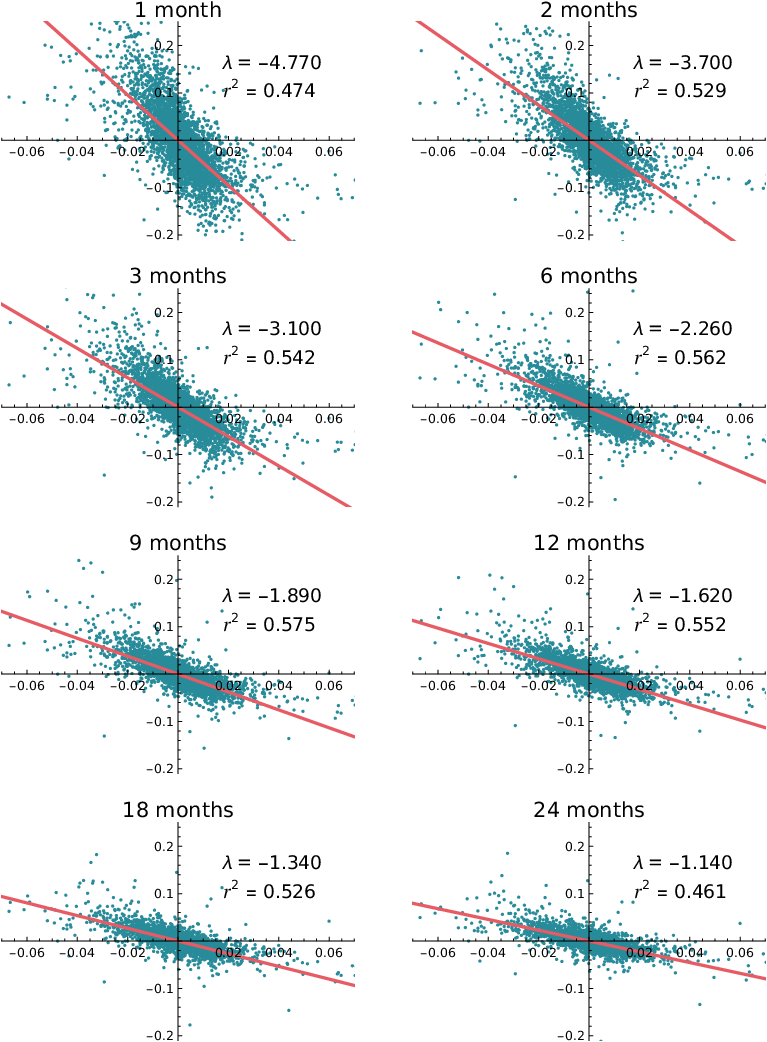}
 \caption{Regression of S\&P 500 daily returns ($x$ axis) vs variance swap returns ($y$ axis)
  for various maturities during the period from 2008 to 2025}
 \label{fig:lambdaregression}
\end{figure}

\begin{table}[h]
\setlength{\tabcolsep}{1.5pt}   \centering
 \begin{tabular}{|l|*{8}{S[table-format=-4.2]|}}
  \hline
  \text{\bf Term (mo)} & {1} & {2} & {3} & {6} & {9} & {12} & {18} & {24} \\
  \hline
  \text{\bf Elasticity } $\lambda$ & -4.76 & -3.69 & -3.09 & -2.26 & -1.89 & -1.61 & -1.34 & -1.13 \\
  \hline
  \text{\bf R-squared} & 0.47 & 0.53 & 0.54 & 0.56 & 0.57 & 0.53 & 0.46 & 0.55 \\
  \hline
  \text{\bf Std. error} & 0.07 & 0.05 & 0.04 & 0.03 & 0.02 & 0.02 & 0.02 & 0.02 \\
  \hline
  \text{\bf F-statistic} & 3903.81 & 4879.56 & 5135.24 & 5562.71 & 5361.64 & 4835.34 & 3723.50 & 5872.21 \\
  \hline
 \end{tabular}
 \caption{Estimates of $\lambda$ for various maturities (2008--2025).}
 \label{t1}
\end{table}

As expected $\lambda$ is negative with an upward sloping term structure. The $R^2$ ranges between 0.46 and 0.57 across all maturities, indicating a moderate fit: daily index returns account for roughly half of the variation in daily variance-swap-rate returns, with the rest left to other factors. The standard error, which measures the sampling variability of the estimated $\lambda$, i.e., how much the coefficient estimate would be expected to fluctuate across repeated samples, and thus how precisely $\lambda$ is pinned down by the data, is small relative to $\lambda$ across all maturities, so the elasticity is estimated fairly precisely throughout.

The F-statistic tests the null hypothesis that the regression has no explanatory power, i.e.\ that $\lambda = 0$. As a rule of thumb, an F-statistic well above $1$ indicates the regression explains meaningfully more variance than noise, while a value close to or below $1$ indicates the model explains no more variance than would be expected by chance. The F-statistic remains consistently large across all maturities, well above the threshold for evidence against the null in every case. This confirms that the regression fit is significant throughout.
\newpage

It is worth noting that the elasticity $\lambda$ appears to be
decreasing proportionally with $1/\sqrt T$ where $T$ is the time to
maturity of the swap. Figure \ref{fig:sqrtcourvefit} shows the values
of $\lambda$ from table \ref{t1} plotted along with $\eta/\sqrt T$,
where $\eta$ was estimated to be $\eta\approx-1.476$ by classical least-squares curve fitting techniques.

\vspace{-0.7cm}
\begin{figure}[H]
 \centering
 \includegraphics[scale=0.7]{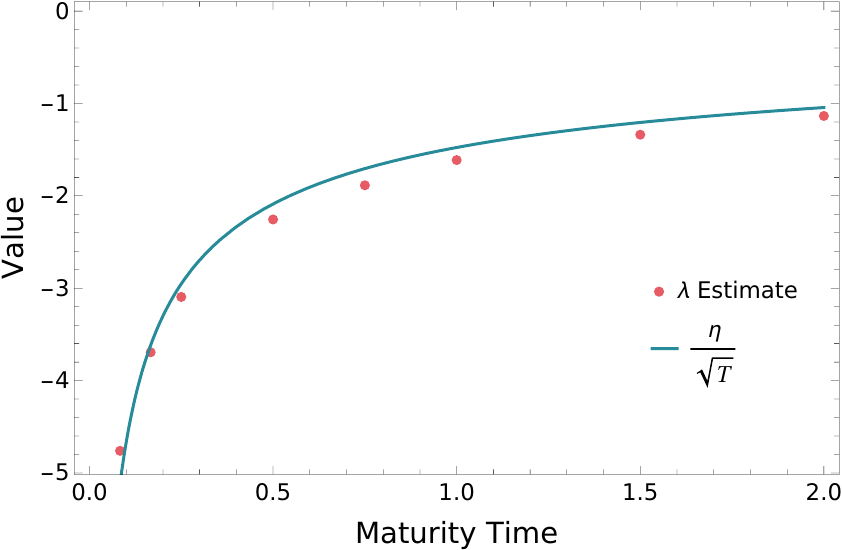}
 \caption{Elasticity of variance swap rate $\lambda$ at different
  maturities (2008--2025). }
 \label{fig:sqrtcourvefit}
\end{figure}

On a different time period, namely from 2022 to 2025, we performed a
similar regression of variance swap returns against index returns, this
time including a 1-day term variance swap. The results are shown in Table \ref{t1b}.

\begin{table}[h]
\setlength{\tabcolsep}{1.5pt}   \centering
 \begin{tabular}{|>{\bf}l|*{8}{S[table-format=-4.2]|}}
  \hline
  \text{Term} & {\text{1 day}} & {\text{1mo}} & {\text{3mo}} & {\text{6mo}} & {\text{9mo}} & {\text{12mo}} & {\text{18mo}} & {\text{24mo}} \\
  \hline
  Elasticity     & -8.43  & -5.04  & -3.47   & -2.56   & -2.07   & -1.75   & -1.38   & -1.13  \\ \hline
  R-squared      & 0.11   & 0.50   & 0.56    & 0.58    & 0.59    & 0.59    & 0.58    & 0.53   \\ \hline
  Std. error & 0.83   & 0.18   & 0.11    & 0.08    & 0.06    & 0.05    & 0.04    & 0.04   \\ \hline
  F-statistic    & 103.43 & 813.67 & 1045.70 & 1132.40 & 1191.75 & 1170.14 & 1142.46 & 934.69 \\ \hline
 \end{tabular}
 \caption{Estimates of $\lambda$ for various maturities (2022--2025).}
 \label{t1b}
\end{table}

Similar to Table \ref{t1}, the $\lambda$ estimates exhibit an upward sloping term structure. The $R^2$ is in the 0.5--0.6 range for maturities of 1 month through 24 months, indicating a reasonably strong fit at these maturities: daily index returns explain a substantial share of the variation in daily returns of the variance swap rate over this shorter, more recent period. 
The standard error remains small relative to $\lambda$, so the elasticity estimates are precise throughout. The F-statistic remains well above the threshold for evidence against the null across all maturities.
The picture is different for the 1-day maturity: the $R^2$ of 0.11 reflects noise dominating the signal at short horizons. This is echoed by the standard error, which at $0.83$ is by far the largest across maturities, both in absolute terms and relative to $\lambda=-8.43$ itself (about $10\%$ of $|\lambda|$, versus roughly $3$--$5\%$ at the other maturities), confirming that the elasticity is estimated markedly less precisely at this horizon. Nonetheless, the F-statistic of $103.43$ confirms that $\lambda$ remains statistically significant, i.e.\ genuinely nonzero, despite this lower precision.

Remarkably, the $1/\sqrt T$ behavior is again confirmed for the 2022--2025 period, as shown in Figure \ref{fig:sqrtcourvefit2}, 
where the values of $\lambda$ are plotted against $\eta/\sqrt T$
where $\eta \approx -1.59$ is chosen to minimize the mean absolute error.

\begin{figure}[H]
 \centering
 \includegraphics[scale=0.7]{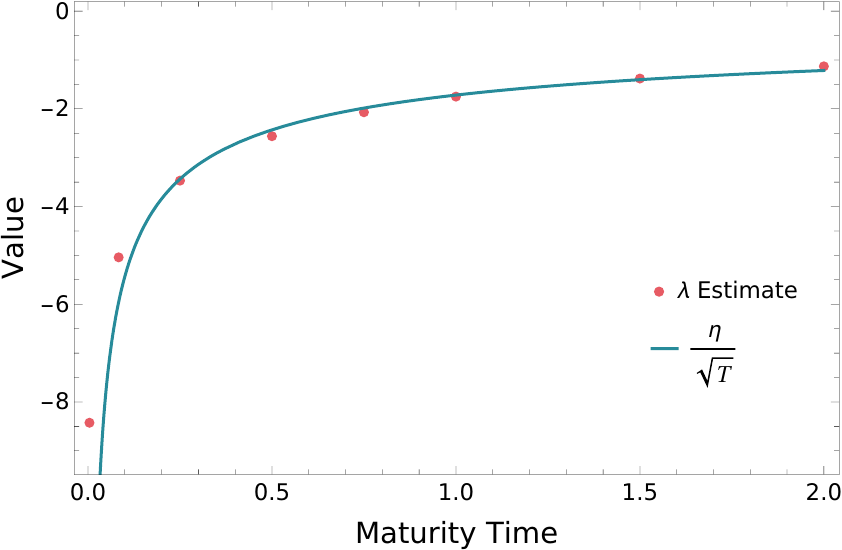}
 \caption{Elasticity of variance swap rate $\lambda$ at different
  maturities (2022--2025). }
 \label{fig:sqrtcourvefit2}
\end{figure}

\vspace{-0.2cm}
We also studied the evolution of $\lambda$ over time by performing a
one year
rolling window regression of variance swap returns against index
returns. The evolution of these $\lambda$ estimates are shown in
Figure \ref{fig:rollinglambda}, and the average values, standard
deviations, average $R^2$, skewness and kurtosis of these rolling
estimates are shown in Table \ref{table:t2}.

\begin{figure}[h]
 \centering
 \includegraphics[scale=0.7]{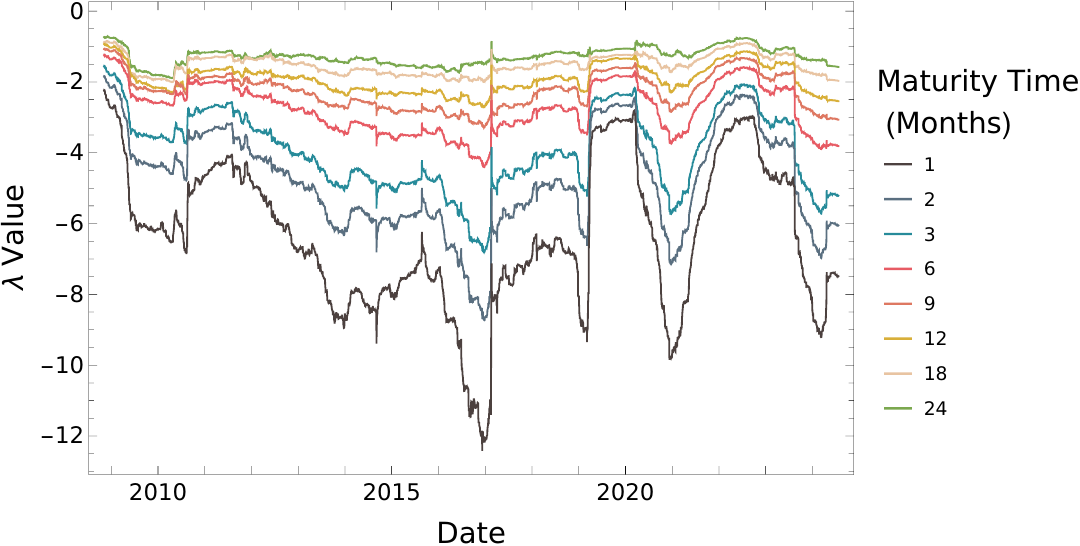}
 \vspace{0.3cm}
 \caption{Rolling estimate of $\lambda$ for various variance swap maturities (2008--2025).}
 \label{fig:rollinglambda}
\end{figure}

\begin{table}[h]
 \centering
 \begin{tabular}{|>{\bf}l|*{8}{S[table-format=-1.3]|}}
  \hline
  \text{Term (months)} & {1} & {2} & {3} & {6} & {9} & {12} & {18} & {24} \\ \hline
  Average $\lambda$ & -6.364 & -4.786 & -3.934 & -2.756 & -2.254 & -1.916 & -1.514 & -1.275 \\ \hline
  Std Deviation & 2.081 & 1.447 & 1.127 & 0.703 & 0.504 & 0.399 & 0.283 & 0.243 \\ \hline
  Average $R^2$ & 0.549 & 0.6 & 0.609 & 0.613 & 0.618 & 0.595 & 0.56 & 0.514 \\ \hline
  Skewness & -0.181 & -0.153 & -0.096 & -0.01 & 0.182 & 0.342 & 0.265 & -0.038 \\ \hline
  Kurtosis & -0.307 & -0.465 & -0.692 & -0.929 & -0.829 & -0.739 & -0.749 & -0.147 \\ \hline
 \end{tabular}
 \caption{Statistics of rolling estimates of $\lambda$ (2008--2025).}
 \label{table:t2}
\end{table}

\newpage
The mean rolling estimates follow the same decreasing-magnitude-with-maturity pattern as the full-sample estimates in Table \ref{t1}, though each rolling mean is somewhat more negative than the corresponding full-sample value. The standard deviation of the rolling $\lambda$ estimates decreases monotonically with maturity, indicating that the elasticity estimate is more stable across rolling windows at longer maturities than at shorter ones. Average $R^2$ is slightly higher throughout than the full-sample $R^2$ in Table \ref{t1}, which confirms that the moderate explanatory power of the regression is stable across rolling windows rather than driven by a few years of data. Skewness is small and changes sign across maturities, indicating the rolling $\lambda$ estimates are close to symmetrically distributed around their mean at every maturity. Kurtosis is negative at every maturity; negative kurtosis means the distribution of rolling estimates is flatter than a normal distribution, i.e., less prone to extreme outlying values.

The $\eta/\sqrt{T}$ behavior observed in Figure~\ref{fig:sqrtcourvefit} also appears to hold for the rolling estimates of $\lambda$, as shown in Figure~\ref{fig:rollinglambda2A}, where each value of $\lambda$ is scaled by $\sqrt{T}$. Figure \ref{fig:rollinglambda2B} shows the evolution of the resulting $\eta$, obtained by curve-fitting the rolling $\lambda$ estimates within each one-year window to $\eta/\sqrt{T}$.
Table~\ref{table:t3} shows summary statistics of these rolling $\eta$ estimates. The average rolling $\eta$ of $-1.898$ is close to the full-sample value $\eta\approx-1.476$, indicating that the $1/\sqrt{T}$ term-structure relationship is broadly stable over time rather than an artifact of the full-sample fit. 
The standard deviation  is modest relative to the mean, indicating that the rolling $\eta$ estimates do not vary excessively across one-year windows.
Skewness is small in magnitude, indicating that the rolling $\eta$ estimates are close to symmetrically distributed around their mean, while the negative kurtosis indicates a flatter-than-normal distribution with fewer extreme outliers.
Finally, the average $R^2$ of $0.964$ confirms that  the $1/\sqrt{T}$ relationship across maturities fits the rolling $\lambda$ estimates extremely well within each window.

\begin{figure}[H]
 \begin{center}
  \begin{subfigure}{.8\linewidth}  \centering
 \vspace{0.3cm}
 \includegraphics[scale=0.75]{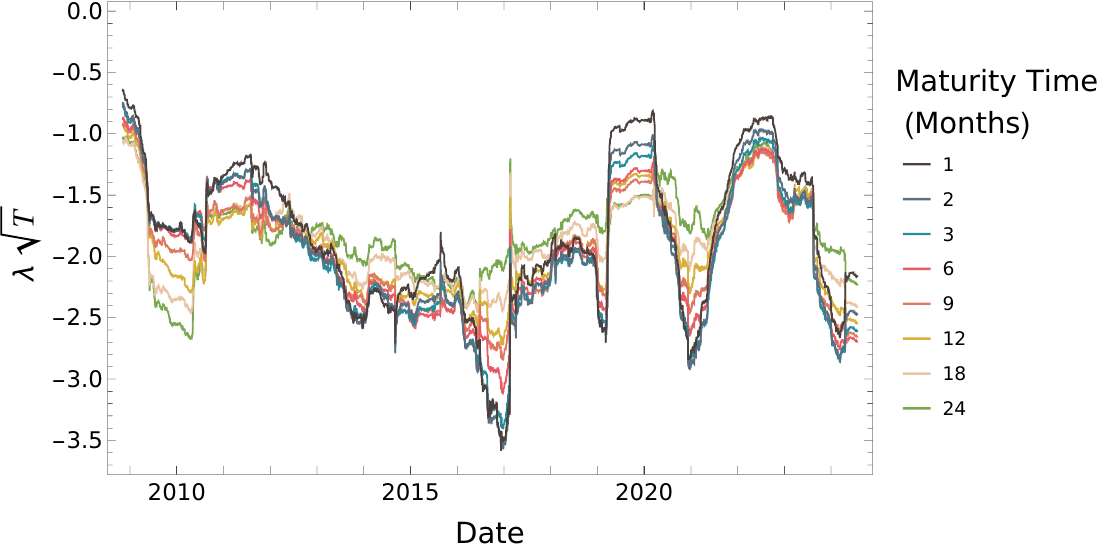}
 \caption{Rolling estimate of $\lambda\sqrt{T}$  (2008--2025).}
 \vspace{0.5cm}
 \label{fig:rollinglambda2A}
  \end{subfigure}
 \end{center}
 \begin{subfigure}{.8\linewidth} \hspace{0.1cm}
 \includegraphics[scale=0.75]{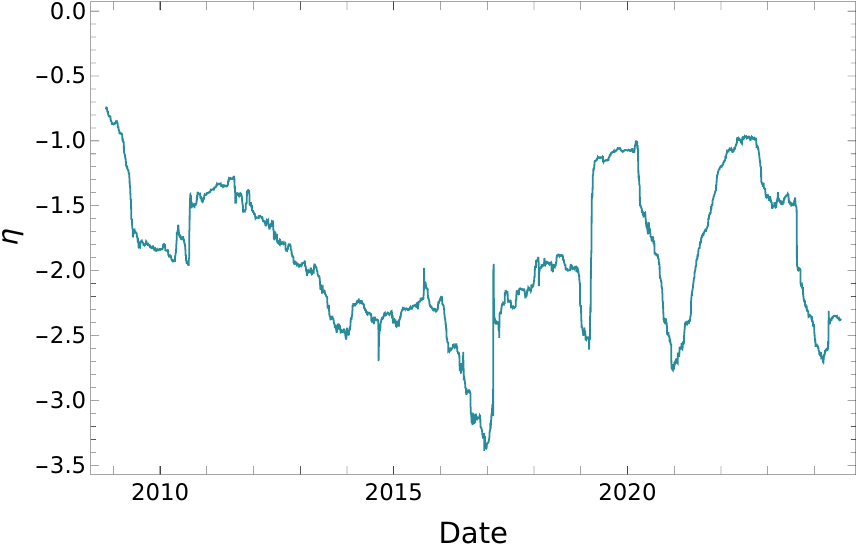}
 \vspace{0.3cm}
 \caption{Rolling estimate of $\eta$ (2008--2025).}
 \label{fig:rollinglambda2B}
 \end{subfigure}
 \vspace{0.3cm}
 \caption{Rolling estimates of $\lambda\sqrt{T}$ and of $\eta$, obtained by fitting $\lambda\approx\eta/\sqrt{T}$ to the rolling $\lambda$ estimates of Figure~\ref{fig:rollinglambda} in each one-year window, over the period 2008--2025.}
 \label{fig:rollinglambda2}
\end{figure}

\begin{table}[h]
 \centering
 \begin{tabular}{|>{\bf}l|S[table-format=-1.3]|}
  \hline
  Average $\eta$ & -1.898 \\ \hline
  Std. deviation & 0.552 \\ \hline
  Average $R^2$ & 0.964 \\ \hline
  Skewness & -0.081 \\ \hline
  Kurtosis & -0.528 \\ \hline
 \end{tabular}
 \vspace{0.3cm}
 \caption{Statistics of the rolling estimates of $\eta$ (2008--2025).}
 \label{table:t3}
\end{table}

\newpage\subsection{Option power delta}

In addition to the 1-month SVI fit of Section~\ref{sec:svitest}, we also performed a 12-month smile calibration as of August 20, 2025.  Table~\ref{tab:svi1y} reports the calibrated SVI parameters for the 12-month maturity. The resulting at-the-money-forward volatility was 15.6\%, and the fit quality, a relative RMSE of 2.1\% against the market implied volatilities, indicates a good match to the observed smile.
\vspace{0.5cm}
\begin{table}[h]
\centering
\begin{tabular}{>{\bf}lccccc}
\toprule
Parameter & $a$ & $b$ & $\rho$ & $s$ & $m$ \\
\midrule
Value & $-0.131170$ & $0.249869$ & $-0.108580$ & $0.581558$ & $0.167150$ \\
\bottomrule
\end{tabular}
\caption{SVI calibration for the 12-month SPX smile, August 20, 2025.}
\label{tab:svi1y}
\end{table}

Based on the 1-month and 12-month SVI fits, we calculated the hedge ratios generated by three methods: Black-Scholes delta, SVI total delta, and power smile delta, using relevant $\lambda$ estimates from Table~\ref{t1} ($\lambda = -1.61$ for 12-month and $\lambda = -4.76$ for 1-month) and Table~\ref{t1b} ($\lambda = -1.75$ for 12-month and $\lambda = -5.04$ for 1-month). Figure~\ref{fig:calldeltas} plots our results for 12-month and 1-month call options as functions of log-moneyness $x = \ln(K/F)$. We can see that the ``power delta'' is always lower than the ``powerless'' total delta, as predicted in Section~\ref{sec:powersmile}, and in line with the Black-Scholes delta.  Figure~\ref{fig:calldeltas2} repeats this hedge ratio computation using the full-sample (2008--2025) $\lambda$ estimates of Table~\ref{t1}.

As mentioned in introduction, there is empirical evidence that hedge ratios should be lower than the Black-Scholes delta.  This may be achieved using lower values of $\lambda$, for example the mean value minus two standard deviations based on Table~\ref{table:t2} rolling estimate statistics.  This would give $\lambda = -1.916 - 2\times0.399 = -2.71$ for 12-month and $\lambda = -6.364 - 2\times 2.081 = -10.53$ for 1-month.  Figure  \ref{fig:calldeltas3} shows the corresponding results.  We can see that this methodology indeed results in a power delta that is lower than the Black-Scholes delta.

\newpage
\begin{figure}[H]
 \centering
 
 \begin{subfigure}{.8\linewidth}  \centering
  \includegraphics[width=\textwidth]{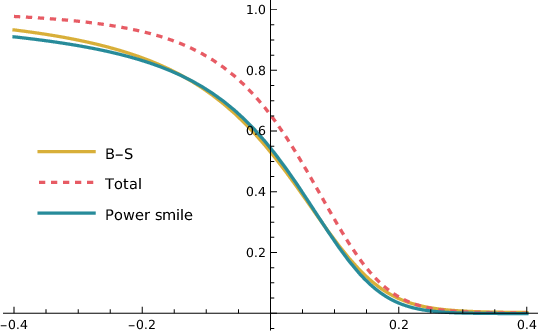}    \caption{12-month calls, $\lambda = -1.75$ from Table \ref{t1b}}
 \end{subfigure}
 \begin{subfigure}{.8\linewidth}  \centering
  \includegraphics[width=\textwidth]{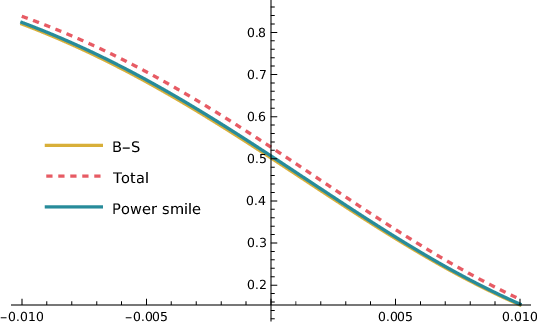}    \caption{1-month calls, $\lambda = -5.04$ from Table \ref{t1b}}
 \end{subfigure}
  \caption{Initial call deltas as functions of log-moneyness $x = \ln(K/F)$ }
  \label{fig:calldeltas}
 \begin{flushleft}
    \footnotesize     \textbf{Note:} Unlike call delta plots by spot price commonly found in the option literature, here each point on a curve is the delta  of a different call option.  The first quadrant are deltas of out-of-the-money calls and the second quadrant are deltas of in-the-money calls.
 \end{flushleft}
\end{figure}

\begin{figure}[H]
 \centering

 \begin{subfigure}{.8\linewidth}  \centering
  \includegraphics[width=\textwidth]{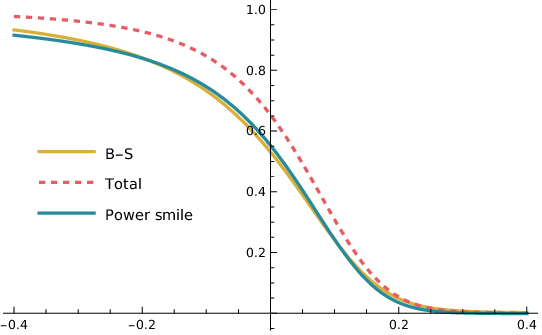}    \caption{12-month calls $\lambda = -1.61$ from Table \ref{t1}}
 \end{subfigure}
 \begin{subfigure}{.8\linewidth}  \centering
  \includegraphics[width=\textwidth]{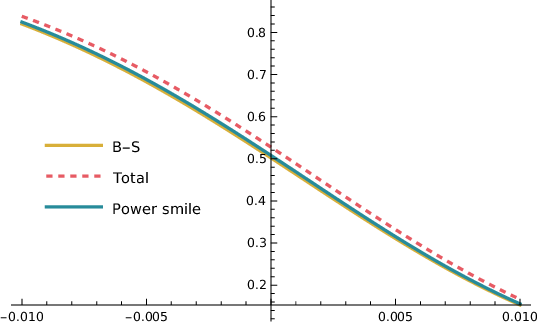}    \caption{1-month calls, $\lambda = -4.76$ from Table \ref{t1}}
 \end{subfigure}
  \caption{Initial call deltas as functions of log-moneyness $x = \ln(K/F)$}
   \label{fig:calldeltas2}
 \begin{flushleft}
    \footnotesize     \textbf{Note:} Unlike call delta plots by spot price commonly found in the option literature, here each point on a curve is the delta of a different call option.  The first quadrant are deltas of out-of-the-money calls and the second quadrant are deltas of in-the-money calls.
 \end{flushleft}
\end{figure}

\begin{figure}[H]
 \centering

 \begin{subfigure}{.8\linewidth}  \centering
  \includegraphics[width=\textwidth]{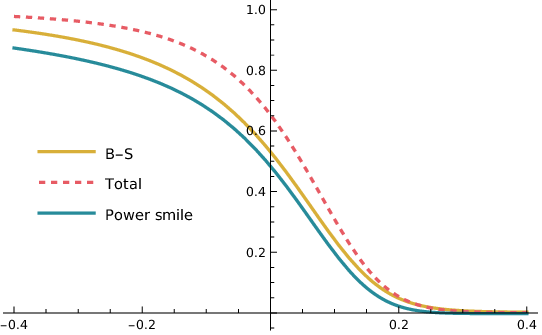}    \caption{12-month calls $\lambda = -2.71$ using Table \ref{table:t2} }
 \end{subfigure}
 \begin{subfigure}{.8\linewidth}  \centering
  \includegraphics[width=\textwidth]{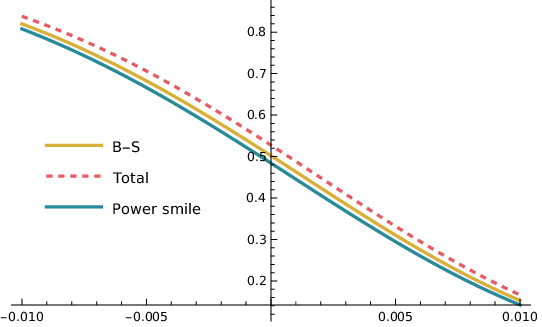}    \caption{1-month calls, $\lambda = -10.53$ using Table \ref{table:t2}}
 \end{subfigure}
  \caption{Initial call deltas as functions of log-moneyness $x = \ln(K/F)$}
   \label{fig:calldeltas3}
 \begin{flushleft}
    \footnotesize     \textbf{Note:} Unlike call delta plots by spot price commonly found in the option literature, here each point on a curve is the delta of a different call option.  The first quadrant are deltas of out-of-the-money calls and the second quadrant are deltas of in-the-money calls.
 \end{flushleft}
\end{figure}

\appendix

\renewcommand{\thesection}{\Alph{section}}
\renewcommand{\theequation}{\thesection\arabic{equation}}

\theoremstyle{theorem}
\newtheorem{lemma}{Lemma}[section]
\expandafter\let\csname corollary\endcsname\relax
\expandafter\let\csname c@corollary\endcsname\relax
\newtheorem{corollary}[lemma]{Corollary}
\expandafter\let\csname remark\endcsname\relax
\expandafter\let\csname c@remark\endcsname\relax
\theoremstyle{remark}
\newtheorem{remark}[lemma]{Remark}

\newpage
\section{Matytsin's formula}\label{sec:appendixa}

\begin{lemma}\label{lem:carr-lee-matytsin}
The fair value of variance may be calculated as
\begin{equation}\label{eq:carr-lee-polishchuk}
  V^*(F) = \int_{-\infty}^{\infty}n\big(d_2(x,F)\big)\sigma^2(x, F)\left(-\frac{\partial d_2}{\partial x}(x,F)\right)\d x,
  \quad d_2(x,F) = \frac{-x-\sigma^2(x, F) T/2}{\sigma(x,F) \sqrt T}.
 \end{equation}
\end{lemma}
\begin{proof}
    See \citet[Sec.~5]{Carr2009}, who credit Andrew Matytsin (private communication), and Alexey Polishchuk for the derivation.  See also \citet[pp.~139--140]{Gatheral2004}.
\end{proof}
\begin{corollary}[Matytsin's formula]
    If $d_2(x,F)$ is strictly monotonically increasing in $x$, we obtain by change of variable $y = d_2(x,F)$  the elegant formula
    \begin{equation}    \label{eq:matytsin}
  V^*(F) =\int_{-\infty}^{\infty}{n\left(y\right){\tilde\sigma}^2(y, F)\d y},
 \end{equation}
    where $\tilde\sigma (y,F)=\sigma(d_2^{-1} (y; F), F)$.
\end{corollary}
\begin{remark}
    If the smile $\sigma(x,F) \equiv \sigma(x)$ is a pure function of (log) moneyness, then equations \eqref{eq:carr-lee-polishchuk} and \eqref{eq:matytsin} do not depend on $F$. This provides an alternative proof that the variance swap total delta is zero.
\end{remark}

\end{document}